\documentclass[reprint,
superscriptaddress,
nofootinbib,
aps,
pra,
showkeys
]{revtex4-2}
\pdfoutput=1

\usepackage{amsthm, mathtools, amssymb}
\usepackage{braket, physics} 
\usepackage{nicefrac}

\usepackage{orcidlink}
\graphicspath{ {./images/} }

\usepackage{graphicx, placeins, float}
\usepackage{multirow, makecell, hhline} 
\usepackage[caption=false]{subfig}
\usepackage[export]{adjustbox}

\usepackage{xcolor}

\usepackage{hyperref}
\usepackage[nameinlink,capitalize]{cleveref}
\hypersetup{
	colorlinks = true,
	linkbordercolor = {white},
	linkcolor={purple},
	citecolor={purple},
	urlcolor={blue}
}

\Crefname{appendix}{Supplementary Material}{Supplementary Materials}

\DeclareMathOperator*{\argmin}{arg\,min}

\newcommand*\dif{\mathop{}\!\mathrm{d}}

\newtheorem*{theorem*}{Theorem}
\newtheorem{proposition}{Proposition}

\newcommand{\I}{\mathcal{I}}
\newcommand{\U}{\mathcal{U}}
\newcommand{\f}{f}
\newcommand{\F}{F}
\renewcommand{\O}{\mathcal{O}}

\begin{document}
\title{Quantum optical classifier with superexponential speedup}
\author{Simone Roncallo\,\orcidlink{0000-0003-3506-9027}}
	\email[Simone Roncallo: ]{simone.roncallo01@ateneopv.it}
	\affiliation{Dipartimento di Fisica, Università degli Studi di Pavia, Via Agostino Bassi 6, I-27100, Pavia, Italy}
	\affiliation{INFN Sezione di Pavia, Via Agostino Bassi 6, I-27100, Pavia, Italy}
	
\author{Angela Rosy Morgillo\,\orcidlink{0009-0006-6142-0692}}
	\email[Angela Rosy Morgillo: ]{angelarosy.morgillo01@ateneopv.it}
	\affiliation{Dipartimento di Fisica, Università degli Studi di Pavia, Via Agostino Bassi 6, I-27100, Pavia, Italy}
	\affiliation{INFN Sezione di Pavia, Via Agostino Bassi 6, I-27100, Pavia, Italy}
	
\author{Chiara Macchiavello\,\orcidlink{0000-0002-2955-8759}}
	\email[Chiara Macchiavello: ]{chiara.macchiavello@unipv.it}
	\affiliation{Dipartimento di Fisica, Università degli Studi di Pavia, Via Agostino Bassi 6, I-27100, Pavia, Italy}
	\affiliation{INFN Sezione di Pavia, Via Agostino Bassi 6, I-27100, Pavia, Italy}
	
\author{Lorenzo Maccone\,\orcidlink{0000-0002-6729-5312}}
	\email[Lorenzo Maccone: ]{lorenzo.maccone@unipv.it}
	\affiliation{Dipartimento di Fisica, Università degli Studi di Pavia, Via Agostino Bassi 6, I-27100, Pavia, Italy}
	\affiliation{INFN Sezione di Pavia, Via Agostino Bassi 6, I-27100, Pavia, Italy}
	
\author{Seth Lloyd\,\orcidlink{0000-0003-0353-4529}}
	\email[Seth Lloyd: ]{slloyd@mit.edu}
	\affiliation{Massachusetts Institute of Technology, Cambridge, MA 02139, USA}
	
\begin{abstract}
	Classification is a central task in deep learning algorithms. Usually, images are first captured and then processed by a sequence of operations, of which the artificial neuron represents one of the fundamental units. This paradigm requires significant resources that scale (at least) linearly in the image resolution, both in terms of photons and computational operations. Here, we present a quantum optical pattern recognition method for binary classification tasks. It classifies objects without reconstructing their images, using the rate of two-photon coincidences at the output of a Hong-Ou-Mandel interferometer, where both the input and the classifier parameters are encoded into single-photon states. Our method exhibits the behaviour of a classical neuron of unit depth. Once trained, it shows a constant $\mathcal{O}(1)$ complexity in the number of computational operations and photons required by a single classification. This is a superexponential advantage over a classical artificial neuron.
\end{abstract}
\keywords{Quantum classifier; Quantum optical neuron; Quantum neural networks; Hong-Ou-Mandel effect;}
\maketitle

\section{INTRODUCTION}
Image classification has been significantly fostered by the introduction of deep learning methods, which provide several algorithms that can learn and extract image features. Examples include feedforward neural networks, convolutional neural networks and vision transformers \citep{art:LeNet,art:AlexNet,art:ResNet,art:ViT}. The artificial neuron, also called perceptron \citep{art:Rosenblatt}, represents the fundamental unit of such architectures. In this model, encoded data are processed through a set of weighted trainable connections, by taking the scalar product between the input and the vector of weights. The output is further post-processed, including a bias and an activation function, which is usually non-linear \citep{book:Goodfellow}. Image classification implies a two-fold cost. Computational processing requires a number of operations that scales, at least, linearly in the image resolution. Similarly, the optical cost of image capturing undergoes the same scaling in the number of photons.

When combining multiple neurons, the large number of parameters involved motivates a consistent effort in reducing the cost of deep learning algorithms, e.g. by leveraging classical implementations that bypass hardware in an all-optical way \citep{art:Shastri,art:Lin,art:Zuo,art:Colburn,art:Li,art:Luo,art:McMahon}. Quantum mechanical effects, like superposition and entanglement, can provide a significant speedup in such tasks \citep{art:Mohseni,art:Cai}, e.g. by building quantum analogues of the perceptron \citep{art:Mangini,art:Tacchino,art:Mangini-Tacchino}, by employing variational methods \citep{art:Cerezo1,art:Cerezo2} or quantum-inspired approaches \citep{art:Senokosov}. Quantum optical neural networks harness the best of both worlds, i.e. deep learning capabilities from quantum optics \citep{art:Steinbrecher,art:Killoran,art:Bartkiewicz,art:Sui,art:AonanZhang,art:Stanev,art:Wood}.

In this paper, we introduce a quantum optical setup to classify objects without reconstructing their images. Our approach relies on the Hong-Ou-Mandel effect, for which the probability that two photons exit a beam splitter in different modes, depends on their distinguishability \citep{art:Mandel,art:Garcia-Escartin,art:Sadana,art:Hiekkamaki}. In our implementation, an input object is targeted by a single-photon source, and eventually followed by an arbitrary lens system. The single-photon state interferes with another one, which encodes a set of trainable parameters, e.g. through a spatial light modulator. After the Hong-Ou-Mandel interferometer, the photons are collected by two bucket detectors without spatial sensitivity, one for each output mode. Classification occurs by measuring the rate of two-photon coincidences at the output (see \cref{fig:Setup}).
\begin{figure}[ht]
	\centering
	\includegraphics[width = 0.45 \textwidth]{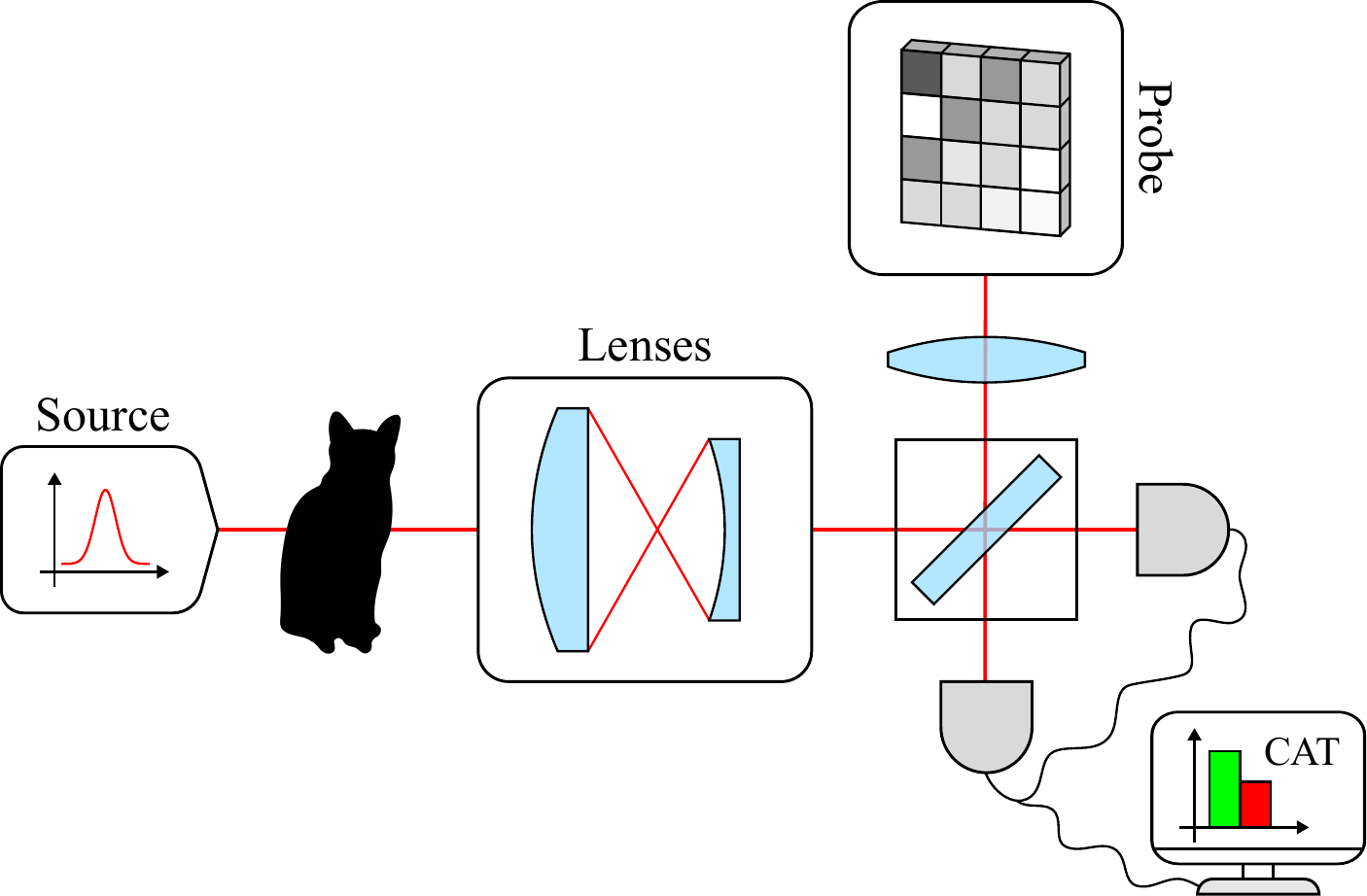}%
	\caption{\label{fig:Setup}Quantum optical neuron implemented by the Hong-Ou-Mandel interferometer. An object is targeted by a single-photon source and classified through the rate of two-photon coincidences at the interferometer output, obtained by placing two photodetectors with no spatial resolution, i.e. without reconstructing the object image. In the top branch, an additional thin lens can translate the classification problem to the Fourier domain.}
\end{figure}

The Hong-Ou-Mandel effect has been successfully applied to quantum kernel evaluation \citep{art:Bowie}, which can compute distances between pairs of data points in the feature space. In this case, each point is sent to one branch of the interferometer, encoded in the temporal modes of a single-photon state. In our method, the interferometer has only one independent branch, which takes the spatial modes of a single-photon state reflected off the target object. The other branch remains fixed after training, and contains the layer of parameters. After the measurement, the response function of our apparatus mathematically resembles that of a classical neuron. For this reason, we refer to our setup as quantum optical neuron. By analytically comparing the resource cost of the classical and quantum neurons, we show that our method requires constant $\mathcal{O}(1)$ computational operations and injected photons, whereas the classical methods are at least linear in the image resolution: a superexponential advantage.

\section{METHOD}
In this section, we discuss the apparatus of \cref{fig:Setup}, without explicitly modelling the probe. Two single-photon states are fed into the left and top branches of a $50\!:\!50$ beam splitter, acting as input and processing layers, respectively. In the left branch, the single-photon source reflects off the object, and reaches the beam splitter after a linear optical system. In the top branch, we consider a generic single-photon state, which depends on a set of trainable real parameters. We count the two-photon coincidences at the beam splitter output. We show how to interpret the Hong-Ou-Mandel response as the one produced by a single-layer neural network-like operation on the object image.

We call input and probe modes, i.e. $a$ and $b$, those fed into the left and top branches of the interferometer. In the input branch, a single photon with spectrum $\phi$ is generated at the longitudinal origin $z=0$, followed by an object with two-dimensional shape $\O$. An imaging system with transfer function $\mathcal{L}_d$, e.g a pinhole or a linear optical apparatus, is placed after the object. Here, $z_o$ and $z_i$ are the longitudinal positions of the object and the image plane, respectively, and $d = z_i - z_o$ their displacement.

The output of the imaging optics reads (see Supplementary Material A)
\begin{equation}
    \ket{\Psi_\I} = \int \dif^2k \ \hat{\I}_\omega(k | \O) a^\dagger_\omega(k) \ket{0} \ ,
    \label{eq:PsiO}
\end{equation}
with $\hat{\I}_\omega(\cdot|\O) = [(\hat{\phi}_\omega \hat{\mathfrak{H}}_{z_o})*\hat{\O}] \hat{\mathcal{L}}_{d}$ the total transfer function from the single-photon source to the image plane, and $a^\dagger_\omega(k)$ the creation operator of a photon in the input mode, acting on the vacuum state $\ket{0}$. The hat operator denotes the two-dimensional Fourier transform on the transverse coordinates plane, $*$ the convolution operation, $\mathfrak{H}_{z_o}$ the transfer function from the source to the object plane, $k = (k_x, k_y)$ the transverse momentum, and $\omega$ the frequency conjugated to the temporal degree of freedom of the electromagnetic potential.

In the probe branch, a generic quantum state is prepared, eventually followed by a linear optical system. At the beam splitter plane, the probe state reads
\begin{equation}
    \ket{\Psi_\U} = \int \dif^2k \ \hat{\U}_\omega(k|\lambda) b_\omega^\dagger(k) \ket{0} \ ,
    \label{eq:PsiProbe}
\end{equation}
with $\lambda = \{\lambda_{i_1 \ldots i_n}\}$ a collection of (trainable) parameters, $\U$ the spatial spectrum of the probe, and $b^\dagger_\omega(k)$ the creation operator of a photon in the probe mode.

A photodetector with no spatial resolution is placed at the output of each branch. After feeding both states into a $50\!:\!50$ beam splitter, the rate of two-photon coincidences reads
\begin{align}
    & p(1_a \cap 1_b|\lambda, \O) = \frac{1}{2}\left[\alpha_\lambda(\O) - \f_\lambda(\O) \right] \ ,
    \label{eq:BosonicCoincidences} \\
    & \text{with} \ \ \begin{aligned}
		\alpha_\lambda(\O) &= || \I_\omega(\cdot|\O) ||^2 || \U_\omega(\cdot|\lambda) ||^2  \ , \\
		\f_\lambda(\O) &=  \left| \langle \I_\omega(\cdot|\O),\U_\omega (\cdot|\lambda) \rangle \right|^2 \ ,
	\end{aligned}
\end{align}
where $|| \cdot ||$ and $\langle \cdot , \cdot \rangle$ denote the $L^2$-norm and inner product, respectively. Here, $\alpha_\lambda(\O)$ depends on the normalization of the input and probe states, which can be $\alpha_{\lambda} < 1$ in the presence of optical losses. Whenever the two spectra are indistinguishable, i.e. when $\U$ perfectly matches $\I$, coincidences are not observed. On the other hand, the more distinguishable the input and the probe states are, the smaller $\langle \I(\cdot|\O),\U(\cdot|\lambda) \rangle$ becomes and the rate of coincidences increases. See Supplementary Material B for a derivation, and Supplementary Material D for a similar result in the Fourier domain.

At the image plane $I$, with transverse coordinates $r = (x,y)$, we have
\begin{equation}
	 \f_\lambda(\O) = \left| \int_I \dif^2 r \ \I_\omega(r|\O)\U_{\omega}^*(r|\lambda) \right|^2 \ .
	 \label{eq:NetworkBraketPosition}
\end{equation}
This integral measures the point-wise overlap between the input image and the probe. We interpret it as the prediction of our classification model, where $\f_\lambda \in [0,1]$ represents the probability that $\I$ belongs to the class of $\U$. In particular, $\f_\lambda \to 0$ ($\f_\lambda \to 1$) when the class of $\I$ is orthogonal to (is the same of) $\U$. In the next section, we show how to encode a generic class in $\U$, by means of the optimization of the set of parameters $\lambda$.

The output measurement introduces a non-linear operation after the beam splitter, represented by the squared absolute value in the left-hand side of \cref{eq:NetworkBraketPosition}. We increase the predictability of our model, by enhancing this non-linearity through the following post-processing operations. Consider the sigmoid (logistic) function
\begin{equation}
	\sigma(x) := \frac{1}{1+e^{-\beta x + \gamma}} \ ,
	\label{eq:SigmoidFunction}
\end{equation}
where $\beta, \gamma$ are hyperparameters, i.e. constants with respect to the training process. We introduce an additional trainable parameter $b \in \mathbb{R}$, called bias, which, combined with $f_{\lambda}$ and $\sigma$, yields
\begin{equation}
	\F_{b\lambda}(\O) = \sigma(\f_{\lambda}(\O) + b) \ ,
	\label{eq:FinalOutputPosition}
\end{equation}
which determines the label predicted by the Hong-Ou-Mandel apparatus. These modifications can improve the performance of the neuron. The sigmoid increases the non-linearity introduced by the squared absolute value, and so the predictability of the model. In addition, the bias is introduced on heuristic motivations: it compensates the constraint given by the normalization in \cref{eq:BosonicCoincidences}, while enhancing the robustness of our protocol against optical losses (which may affect the above-mentioned normalizability, yielding $\alpha_{\lambda} < 1 $).

We now discuss the training stage. Consider a training set, i.e. an ensemble of objects $\{\O_j\}$ with target labels $\{y_j \in \{0,1\}\}$. We separately feed each object into the input branch of the interferometer. Predicted and target classes are compared in terms of their binary cross-entropy, which is used as loss function of a gradient descent optimizer. The optimizer updates $\lambda$ through the derivative of the loss function, whose only model-dependent contribution is
\begin{equation}
	\partial_\lambda \f = 2 \Re \left[ \langle \I_{\omega}, \U_{\omega} \rangle \langle \I_{\omega}, \partial_{\lambda}\U_{\omega} \rangle^* \right] \ .
	\label{eq:DerivativeComplete}
\end{equation}	
Ideally, the training is complete after finding a set of parameters  that minimizes the loss. When the computation of the loss function derivative is not possible, we can replace the gradient descent method with the coordinate descent one. For each epoch, the loss function is evaluated on a discrete neighbourhood $\Lambda$ in the parameters space. The update rule is
\begin{equation}
	\lambda \to \argmin_{\lambda \in \Lambda} H(\lambda) \ . 
\end{equation}
Notice that our model is resilient against the issue of gradient explosion \citep{art:Glorot}, since it depends on physical data and functions only. See Supplementary Material C for a discussion. 

There is a formal relationship between the post-processed output of the Hong-Ou-Mandel interferometer of \cref{eq:FinalOutputPosition} and that of a classical neuron. Consider $f_{\lambda}(\O)$ discretized and vectorized in a mesh of $N$ cells, either in the spatial or in the Fourier domain. Then, \cref{eq:FinalOutputPosition} corresponds to the composition of a real-valued neuron, with $N$ trainable weights, square absolute value activation function and no bias, and a second neuron, with a scalar unit weight, sigmoid activation function and a trainable bias. Namely
\begin{equation}
   G_{bw}(x) = \sigma\big(|w \cdot x|^2 + b \big) \ ,
   \label{eq:PerceptronAnalogy}
\end{equation}
where $x \in \mathbb{C}^{N}$ is the input, while $w \in \mathbb{C}^{N}$ and $b \in \mathbb{R}$ are the weights and bias, respectively. We can formally identify $G_{bw}(x)$ with $F_{b\lambda}(\O)$ under the substitution
\begin{equation}
	\left( x, w \right) \xleftarrow{\sim} \big( \I_\omega (r|\O),  \U_\omega (r|\lambda) \big)  \ ,
\end{equation}
where $\xleftarrow{\sim}$ is the discretization and vectorization to $\mathbb{C}^{N}$. This analogy is represented in \cref{fig:Perceptron}.
\begin{figure}
    \centering
    \includegraphics[width = 0.425 \textwidth]{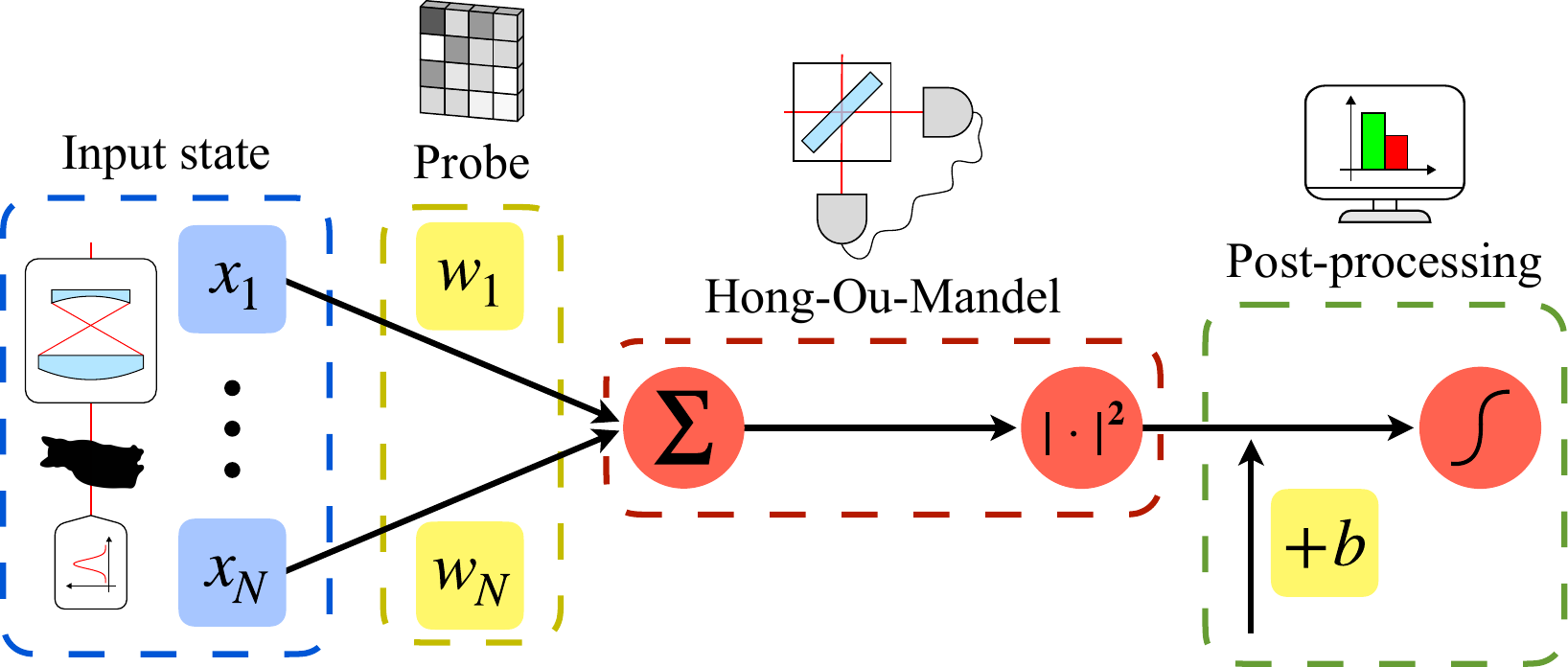}
    \caption{Mathematical equivalence between the Hong-Ou-Mandel and a classical artificial neuron. The left branch of the interferometer corresponds to the input layer, while the probe parameters are related to the trainable neuron weights. The rate of coincidences encodes the square absolute value of their scalar product, further post-processed by adding a bias and a sigmoid activation function.}
    \label{fig:Perceptron}
\end{figure}

\begin{table}[H]
	\centering
	\def\arraystretch{1.5}
	\setlength\tabcolsep{5pt}
	\begin{tabular}{|c|c|c|c|}
		\hline
		\multicolumn{2}{|c|}{Resources} & \multicolumn{1}{c|}{Quantum} & Classical \\ \hline
		\multicolumn{2}{|c|}{\parbox{3cm}{\centering \ \\[0.5pt] Computational \\ (\# of operations) \\[3pt]}} & $\mathcal{O}(1)$ & $N$ \\ \hline 
 		\multirow{2}{*}{\parbox{2.25cm}{\centering Optical \\ (\# of photons)}} & Imaging & None & $\eta N$ \\ \cline{2-4}
 		& Classification & $\mathcal{O}(1)$ & $\Omega(N)$ \\ \hline
	\end{tabular}
	\caption{\label{tab:ResourceCost}Comparison between the quantum optical neuron and its classical counterpart. Resource cost when reconstructing and classifying an image of $N$ pixels with signal-to-noise ratio $\eta$.  Our method achieves a superexponential speedup: $\Omega(N) \to \mathcal{O}(1)$. See Supplementary Material F for an expanded table.}
\end{table}

\begin{figure*}
	\centering
	\includegraphics[width = 1 \textwidth]{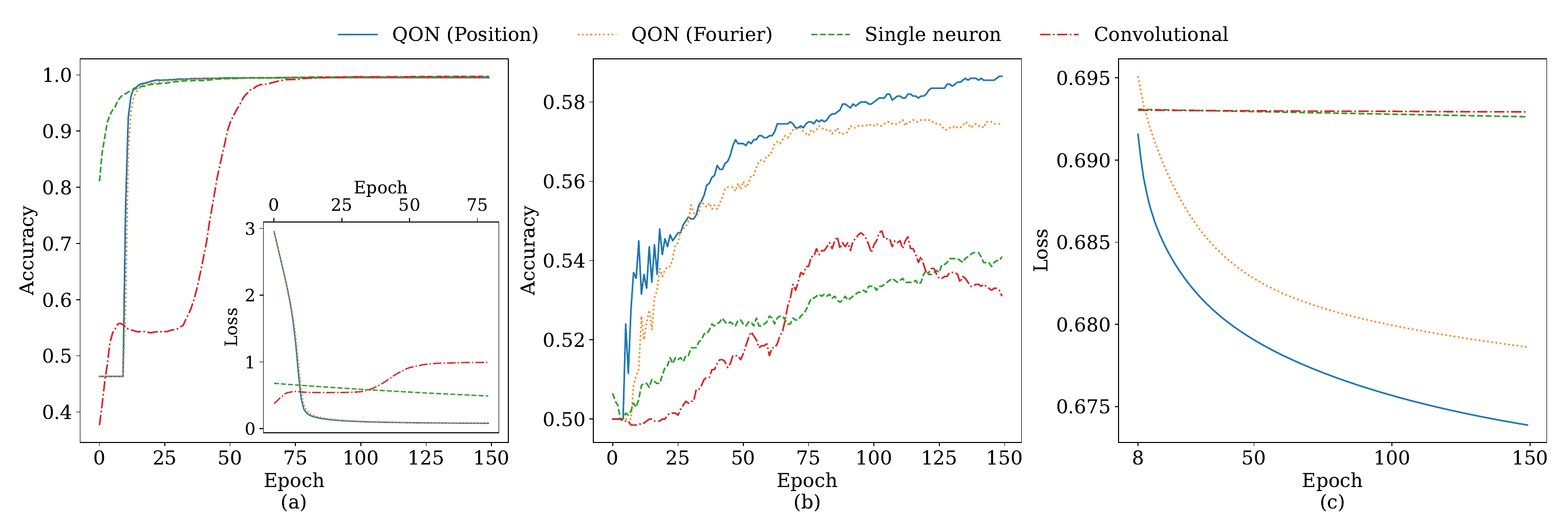}%
	\caption{\label{fig:Training}History plot of a quantum optical neuron, a classical artificial neuron and a convolutional network. All models are trained with the same number of $\sim 1024$ parameters, optimizer and learning rates. The quantum optical neuron (QON) is modelled by an amplitude modulated probe made of $32 \times 32$ pixels, both in the spatial (blue solid line) and in the Fourier (orange dotted line) domains. Comparison with a single artificial neuron (green dashed line) and a convolutional network (red dash-dotted line). The optimization is performed with learning rates $\eta_\lambda = 0.075$ and $\eta_b = 0.005$. (a) Accuracy versus the number of training epochs for the MNIST dataset. The models are trained to distinguish among images of zeros and ones, showing compatible results in terms of trainability and accuracy,  whose final value is above $99\%$. The inset is a history plot of the binary cross-entropy, used as loss function in the gradient descent optimization. (b-c) Accuracy and binary cross-entropy plots versus the number of training epochs for the CIFAR-10 dataset. The models are trained to classify images of cats and dogs. Our method reaches an asymptotic accuracy above $58\%$, showing an advantage with respect to its classical counterparts.}
\end{figure*}

A classical neuron requires at least $N$ photons and $N$ computational operations to classify an image composed of $N$ pixels. Our setup bypasses both costs, by leveraging two essential features. On the one hand, it is completely optical, avoiding the computational need of processing the image. On the other hand, it classifies patterns through the Hong-Ou-Mandel effect, bypassing the photon cost of imaging. In both ways, it provides a superexponential speedup, from $\mathcal{O}(N)$ to $\mathcal{O}(1)$. Parameters can be trained by either classical simulation or direct experimental implementation. Simulating the apparatus would require exponential resources, although reducing the overall experimental effort. Such cost can be mitigated by transferring the parameters from classical models, if available. An experimental approach would guarantee the superexponential speedup even for training, introducing at most an overhead depending on the number of epochs and the pattern complexity. Photon losses due to absorption introduce a constant overhead in both the classical and quantum strategies, which depends on the total reflectivity of the object. We summarize the advantage in \cref{tab:ResourceCost}. See Supplementary Material F for a detailed discussion and derivation, and for an additional comparison with classical interferometric schemes, over which our method still exhibits a superexponential speedup.

Finally, we specialize our discussion by replacing the generic probe state $\U$ with a toy model of an amplitude spatial light modulator (SLM), placed in the top branch of the Hong-Ou-Mandel interferometer, e.g. a liquid crystal display with negligible losses \citep{art:Neff}. Different approaches can be investigated, such as phase-only SLM \citep{art:ZichenZhang}, which may exhibit superior resiliency against losses. Consider a pattern on a greyscale grid with $N$ real amplitudes $\{\lambda_{\mu \nu}\}$. Each pixel, labelled by $(\mu,\nu)$, is represented by an $L \times L$ square with center $r_{\mu \nu} = (\mu + 1/2, \nu + 1/2)L$. Upon an overall parameter-independent normalization, the probe can be approximated as a combination of top-hat functions
\begin{equation}
	\U_{\omega}(r|\lambda) = \sum_{\mu,\nu} u(r-r_{\mu \nu})\frac{\lambda_{\mu \nu}}{||\lambda||} \ ,
	\label{eq:LCD}
\end{equation}
where $ ||\lambda||^2 = \sum_{\mu,\nu}\lambda_{\mu\nu}^2$ and $u(r) := \theta(r + L/2)-\theta(r - L/2)$, with $\theta$ the two-dimensional Heaviside step function. Under this choice, \cref{eq:NetworkBraketPosition} simplifies to
\begin{equation}
	\f_\lambda(\O) = \left| \sum_{\mu,\nu} (u \star  \I_{\omega})(r_{\mu\nu}) \frac{\lambda_{\mu \nu}}{||\lambda||} \right|^2 \ ,
	\label{eq:LCDLayerSpatial}
\end{equation}
where $\star$ is the cross-correlation operation. We introduce a bias and a sigmoid activation function, so that the post-processed output reads $\F_{b\lambda}(\O) = \sigma(\f_{\lambda}(\O) + b)$. Assuming that $\I$ is real, \cref{eq:DerivativeComplete} simplifies to
\begin{equation}
	\partial_{\mu \nu} \f \simeq 2 \frac{\sqrt{f}}{||\lambda||}\left[(u \star  \I_{\omega})(r_{\mu\nu}) - \sqrt{f} \frac{\lambda_{\mu \nu}}{||\lambda||} \right] \ ,
	\label{eq:LCDDerivativePosition}
\end{equation}
with $\partial_{\mu \nu}\f  := \partial f / \partial \lambda_{\mu \nu}$. This expression can be evaluated in an all-optical way, by taking the amplitude measurement of $\I$ directly in the left branch of the interferometer, before the beam splitter. This operation can be done offline, and once per training object. In the next section, we present a simulation of these results, for different choices of the dataset.

\section{RESULTS}
We present a simulation of the quantum optical classifier, comparing its performance against classical neural network architectures. We considered two widely recognized datasets: the MNIST, which contains $28 \times 28$ images of handwritten digits from $0$ to $9$, and the CIFAR-10, comprised of $32 \times 32$ color images distributed across $10$ different classes. A fair comparison is guaranteed by padding the MNIST resolution to $32 \times 32$ pixels, while converting the CIFAR-10 to greyscale (colors can be introduced by replacing the monochromatic source with a multi-wavelength one, or by considering multiple sets of parameters that account for the color decomposition of the image.). Each image is identified with the discretized spectrum of a single-photon state, bypassing the simulation of the imaging apparatus. In practice, this means replacing the target object and the lenses with a second SLM, which displays and encodes the input image. We adopted the binary cross-entropy as loss function, combined with the standard (non-stochastic) gradient descent optimizer, and the accuracy, i.e. the proportion of correct predictions over the total ones, computed on the test dataset, as figure of merit of our results. The training and test datasets are constructed with an approximate ratio of $85\!:\!15$. Simulated results are reported in \cref{fig:Training}, the setup is reported in \cref{fig:Simulation}. All the simulations are run in Python and TensorFlow \citep{soft:TensorFlow}. Our method demonstrates significative performances in both datasets. In the MNIST, it achieves accuracy exceeding $99\%$, when discerning between zeros and ones. In the CIFAR-10, it reaches accuracy above $58\%$, when distinguishing between cats and dogs. This difference reflects the complexity of the two classification tasks. 

We compared our model against conventional classifiers, i.e. a single neuron and a convolutional neural network, commonly employed in pattern recognition tasks \citep{art:Bengio, art:Rosenblatt, art:AlexNet}. Adopting the TensorFlow notation, the convolutional structure is: Conv2D ($10$, $3 \times 3$) $\to$ Conv2D ($4$, $2 \times 2$) $\to$ MaxPooling2D ($2 \times 2$). Roughly, all the architectures have $\sim 10^3$ trainable parameters. The performances are equal in the MNIST dataset, both in terms of trainability and final accuracy. In the CIFAR-10 dataset, our classifier outperforms the conventional ones, showing superior efficiency under a strongly-constrained parameters count. These findings emphasize the competitive accuracy of our method, and also its comparative advantage in pattern recognition tasks with a limited number of parameters.

Finally, we conducted additional tests to evaluate the robustness of our method. First, we performed the above comparisons by removing the post-processing operations, i.e. the bias and the sigmoid, using the square absolute value activation only. Under the same conditions of \cref{fig:Training}, the model maintains its training capabilities with reduced asymptotic accuracy, i.e. $88\%$ and $57\%$ for the MNIST and CIFAR-10 datasets, respectively. Then, we assessed resiliency by introducing sampling noise when estimating the coincidence rates at the Hong-Ou-Mandel output, for $100$ and $1000$ experimental repetitions. In the MNIST dataset, predictability remains unaffected. In the CIFAR-10 dataset, higher noise sensitivity is observed, although with compatible to superior capabilities with respect to its classical counterparts. In both cases, local fluctuations do not affect the global convergence of training. Results are reported in the Supplementary Fig. 1 and Supplementary Fig. 2.
\begin{figure}[th]
	\centering
	\includegraphics[width = 0.375 \textwidth]{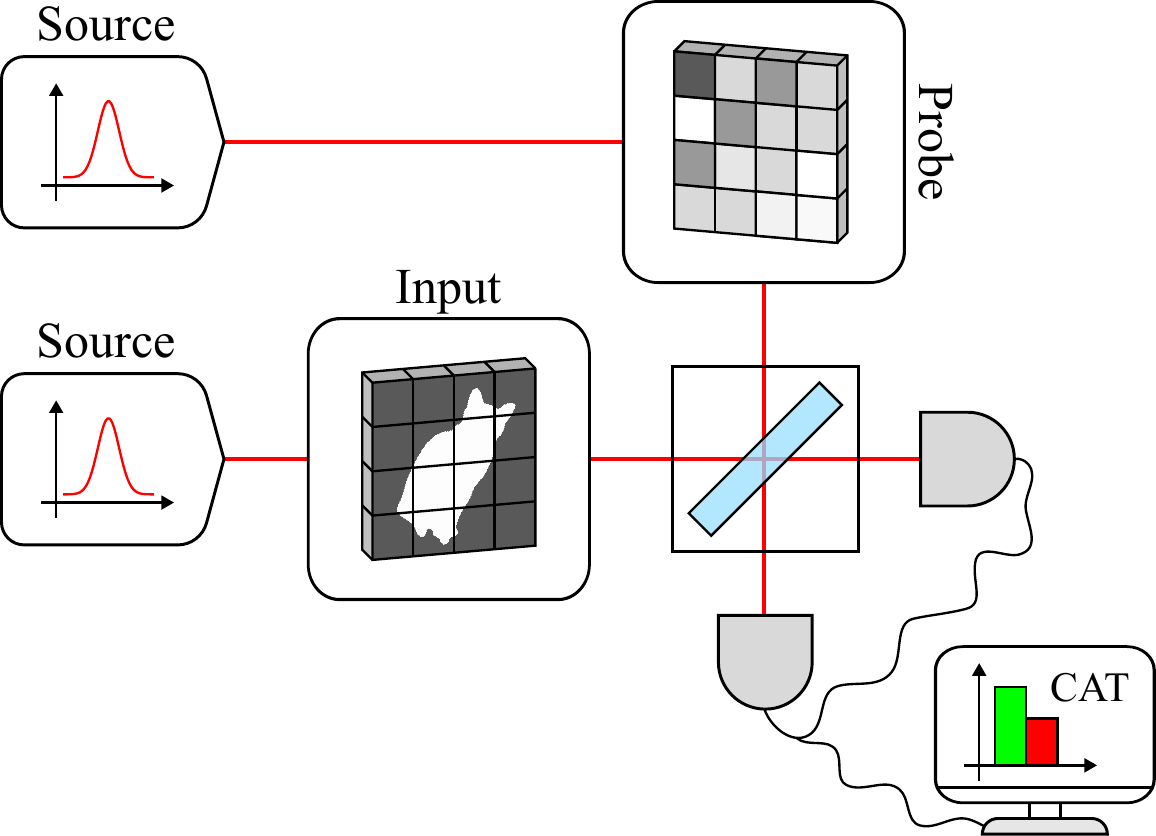}%
	\caption{\label{fig:Simulation}Digital quantum optical neuron. The target object and the optics are replaced by a second spatial light modulator (SLM). The discretized input and spectra correspond to the digital images shown on the respective screen. Both SLMs have the same number of pixels.}
\end{figure}

\section{CONCLUSIONS}
In summary, we introduced an interferometric setup of a quantum optical classifier, with the Hong-Ou-Mandel effect as cornerstone of our classification method. We demonstrated the mathematical relation between our model and a classical neuron, constrained to unit depth, showing their similarity in terms of structure and response function. Its design is completely optical and single-photon based: it provides a superexponential speedup with respect to its classical counterpart, in terms of number of photons and computational resources. Our classifier is inherently binary. However, it can be generalized to multiclass scenarios, adopting one-vs-one or one-vs-rest strategies. In this case, multiple sets of parameters have to be considered to subsequently or adversarially determine the target class. Despite the reduced accuracy against multiclass methods, such approaches would retrieve the original advantage, with an overhead of $\mathcal{O}(K)$ and $\mathcal{O}(K^2)$ operations, respectively, with $K$ the number of classes. After modelling the classifier in terms of a spatial light modulator, we numerically compared our performances against those of standard neural network architectures, showing compatible to superior capabilities in terms of accuracy and training convergence, under the same number of parameters and depending on the pattern complexity. Despite the single neuron limitations in classifying arbitrary complex patterns, our work paves the way for interferometric networks with universal approximation capabilities.

\section*{ACKNOWLEDGEMENTS}
S.R. acknowledges support from the PRIN MUR Project 2022SW3RPY. A.R.M. acknowledges support from the PNRR MUR Project PE0000023-NQSTI. C.M. acknowledges support from the National Research Centre for HPC, Big Data and Quantum Computing, PNRR MUR Project CN0000013-ICSC. L.M. acknowledges support from the PRIN MUR Project 2022RATBS4 and from the U.S. Department of Energy, Office of Science, National Quantum Information Science Research Centers, Superconducting Quantum Materials and Systems Center (SQMS) under Contract No. DE-AC02-07CH11359. S.L. acknowledges support from ARO, DOE, and DARPA.

%
\section*{CODE AVAILABILITY}
The underlying code that generated the data for this study is openly available in GitHub \citep{rep:QON}.

%


\appendix
\renewcommand{\figurename}{SUPPLEMENTARY FIG.} 
\setcounter{figure}{0}   

\renewcommand{\tablename}{SUPPLEMENTARY TABLE.} 
\setcounter{table}{0}   

\begin{widetext}
\section{SINGLE-PHOTON ENCODING\label{app:Imaging}}
In this section, we consider the single-photon state obtained at the output of the left branch of the Hong-Ou-Mandel apparatus, providing a detailed discussion of the input encoding. We adopt units in which $c = 1$.

Consider a generic single-photon state, generated by a monochromatic source with longitudinal position $z$
\begin{equation}
	\ket{\Psi} = \int \dif^3k \ \hat{\Phi}(\mathbf{k}) a^{\dagger}(\mathbf{k})\ket{0} \ ,
\end{equation}
with spectrum $\Phi$ and $ \mathbf{k} = (k_x,k_y,k_z)$. We neglect the polarization of the photon and consider the single-frequency-mode assumption \citep{art:Rezai}, i.e. we assume that the wavefront propagates along definite-sign $z$-directions only. Then, $k = (k_x,k_y)$ represents the only independent degrees of freedom of the single-photon state, which reads
\begin{align}
	\ket{\Psi} &= \int \dif^2k \ \hat{\phi}_\omega(k) a^\dagger_\omega(k)\ket{0} \ , \\
	 &= \int_S \dif^2r \ \phi_\omega(r) a_{\omega}^\dagger(r)\ket{0} \ ,
\end{align}
where $\hat{\phi}_\omega(k) = \hat{\Phi}\left(k_x,k_y,\sqrt{\omega^2-k_x^2-k_y^2}\right)$ and $r = (r_x,r_y)$ labels the transverse coordinates on the source plane $S$.

For simplicity, we assume that the source is placed at the longitudinal origin $z = 0$. Consider an object with two-dimensional shape $\O$, placed at longitudinal position $z_o$. After free-space propagation occurs, the single-photon spectrum undergoes spatial amplitude modulation \citep{book:Saleh}, that is $\Psi_{\O}(r) = \O(r) \Psi(r)_{\to}$, with $\Psi(r)_{\to}$ the spatial input wavefront on the object plane $O$. Namely
\begin{equation}
	\ket{\Psi_{\O}} = \int_O \dif^2r \ [\phi_\omega*\mathfrak{H}_{z_o}](r)\O(r) a^{\dagger}_{\omega}(r)\ket{0} \ ,
	\label{eq:ObjectTransfer}
\end{equation}
where $\mathfrak{H}_{z_o}$ denotes the free-space transfer function between the $S$ and $O$ planes. Using twice the convolution theorem, it follows that 
\begin{equation}
	\ket{\Psi_{\O}} = \int \dif^2k \ [(\hat{\phi}_\omega \mathfrak{H}_{z_o})*\hat{\O}](k) a^{\dagger}_\omega(k) \ket{0} \ .
\end{equation}
Consider a linear optical system with transfer function $\mathcal{L}$, with image plane at longitudinal position $z_i$. By applying again the convolution theorem to $\I_\omega(\cdot|\O) = \left((\phi_\omega*\mathfrak{H}_{z_o}) \O \right)*\mathcal{L}_{z_o - z_i}$, we obtain
\begin{equation}
	\ket{\Psi_\I} = \int \dif^2k \ \hat{\I}_\omega(k | \O) a^\dagger_\omega(k) \ket{0} \ ,
\end{equation}
with $\hat{\I}_\omega(k|\O) = [(\hat{\phi}_\omega \mathfrak{H}_{z_o})*\hat{\O}] \hat{\mathcal{L}}_{d}$. Notice that $\I_\omega(r|\O)$ describes the image formed on a screen placed at distance $d$ from the object.

\section{HONG-OU-MANDEL COINCIDENCES\label{app:Coincidences}}
In this section, we compute the rate of coincidences at the output of the Hong-Ou-Mandel interferometer, using the left and top branch states of the main discussion. We write the input-probe bipartite state as
\begin{equation}
	\ket{\Psi_\I} \otimes \ket{\Psi_\U} = \int \dif^2k_1 \dif^2k_2 \ \hat{\Psi}(k_1,k_2) a^\dagger(k_1) b^\dagger(k_2) \ket{0} \ ,
	\label{eq:PsiBipartite}
\end{equation}
with $\hat{\Psi}(k_1,k_2) = \hat{\I}(k_1 | \O) \hat{\U}(k_2|\lambda)$, where we dropped the $\omega$ subscript for simplicity. The $50\!:\!50$ beam splitter acts as the unitary operation \citep{art:Branczyk}
\begin{equation}
	\begin{cases}
		a^\dagger \to \frac{1}{\sqrt{2}}\left( a^\dagger + b^\dagger \right) \\
		b^\dagger \to \frac{1}{\sqrt{2}}\left( a^\dagger - b^\dagger \right)
	\end{cases} \ ,
	\label{eq:BeamSplitterOperation}
\end{equation} 
yielding
\begin{equation}
	\ket{\Psi_\I} \otimes \ket{\Psi_\U} \to \ket{\Phi} = \frac{1}{2} \int \dif^2k_1 \dif^2k_2 \ \hat{\Psi}(k_1,k_2) \left[ a^\dagger(k_1) + b^\dagger(k_1) \right]\left[ a^\dagger(k_2) - b^\dagger(k_2) \right] \ket{0} \ .
	\label{eq:PsiBipartiteOutput}
\end{equation}
Detection of mode $m \in \{a,b\}$ is described by the projector $\Pi_{m} = \int \dif^2k \ m^\dagger(k)\ket{0}\!\bra{0}m(k)$. The rate of coincidences, i.e. the probability that one and only one photon is detected in each mode, reads
\begin{gather}
	 p(1_a \cap 1_b) = \Tr[\ket{\Phi}\!\bra{\Phi}\Pi_a \otimes \Pi_b] \ , \\
	\text{with} \ \Pi_a \otimes \Pi_b = \int \dif^2k_3 \dif^2k_4 \ a^\dagger(k_3) b^\dagger(k_4) \ket{0}\!\bra{0} a(k_3) b(k_4) \ .
\end{gather}
By substitution of \cref{eq:PsiBipartiteOutput}, we get
\begin{equation}
	p(1_a \cap 1_b) = \frac{1}{4}\int \prod_{i=1}^6 \dif^2k_i \ \hat{\Psi}(k_1,k_2) \hat{\Psi}^*(k_5,k_6) W_1(k_1,k_2,k_3,k_4)W_2(k_3,k_4,k_5,k_6) \ , 
	\label{eq:CoincStep}
\end{equation}
where
\begin{gather}
	\begin{split}
		W_1(k_1,k_2,k_3,k_4) &= \bra{0} a(k_3) b(k_4)\left[a^\dagger(k_1)a^\dagger(k_2) - a^\dagger(k_1)b^\dagger(k_2) + b^\dagger(k_1)a^\dagger(k_2) - b^\dagger(k_1)b^\dagger(k_2) \right] \ket{0} \\
	&= \delta(k_2 - k_3)\delta(k_1 - k_4) - \delta(k_1 - k_3)\delta(k_2 - k_4) \ , 
	\end{split} \\
	\begin{split}
		W_2(k_3,k_4,k_5,k_6) &= \bra{0} \left[a(k_6)a(k_5) - b(k_6)a(k_5) + a(k_6)b(k_5) - b(k_6)b(k_5) \right] a^\dagger(k_3) b^\dagger(k_4) \ket{0} \\
	&= \delta(k_3 - k_6)\delta(k_4 - k_5) - \delta(k_3 - k_5)\delta(k_4 - k_6) \ .
	\end{split}
\end{gather}
By integrating out the Dirac deltas in \cref{eq:CoincStep}, we obtain
\begin{equation}
	p(1_a \cap 1_b) = \frac{1}{2}\int \dif^2k_1 \dif^2k_2 \dif^2k_5 \dif^2k_6 \ \hat{\Psi}(k_1,k_2) \hat{\Psi}^*(k_5,k_6) \left[\delta(k_1 - k_5)\delta(k_2-k_6) - \delta(k_1 - k_6)\delta(k_2 - k_5) \right] \ .
\end{equation}
Finally, the rate of coincidences reads
\begin{equation}
	p(1_a \cap 1_b|\lambda, \O) = \frac{1}{2} \int \dif^2k_1 \ |\hat{\I}(k_1 | \O)|^2 \int \dif^2k_2 \ |\hat{\U}(k_2|\lambda)|^2 - \frac{1}{2} \left| \int \dif^2k \ \hat{\I}(k | \O) \hat{\U}^*(k|\lambda) \right|^2 \ .
\end{equation}
More compactly, 
\begin{equation}
	 p(1_a \cap 1_b|\lambda, \O) = \frac{1}{2}\left[|| \I_\omega(\cdot|\O) ||^2 || \U_\omega(\cdot|\lambda) ||^2 -  \left| \langle \I_\omega(\cdot|\O),\U_\omega (\cdot|\lambda) \rangle \right|^2 \right] \ ,
\end{equation}
with $|| \cdot ||$ and $\langle \cdot , \cdot \rangle$ denoting the $L^2$-norm and inner product.

\section{TRAINING\label{app:Training}}
In this section, we discuss how to train the Hong-Ou-Mandel interferometer as a binary classifier. We separately feed each element of the training set (an ensemble of objects with known labels) into the input branch of the interferometer, comparing the predicted classes with the target ones. We optimize the probe parameters $\lambda$ by means of the gradient descent algorithm, and using the binary cross-entropy as loss function.

Consider a training set made of $M$ objects $\{\O_j\}$, each associated to a binary target label $y_j \in \{0,1\}$, with $0 \leq j \leq M - 1$. We denote $\f^{(j)}_\lambda = \f_\lambda(\O_j)$ our model prediction. After feeding $\O_j$ into the input branch of the interferometer
\begin{gather}
	 \f^{(j)}_\lambda = C - 2p(1_a \cap 1_b|\lambda, \O_j) \ , \\
	 \F_{b\lambda}^{(j)} = \sigma(\f^{(j)}_\lambda + b) \ ,
\end{gather}
where $p \in [0, 1/2]$. For simplicity, we assumed that the losses are independent on both the input and the probe, that is $C := \alpha_\lambda(\O_j) \ \forall \lambda, j$.

Given a sample object, the binary cross-entropy between the target label and the predicted one reads
\begin{equation}
	H\left(y_j,\F_{b\lambda}^{(j)}\right) = - y_j\log (\F_{b\lambda}^{(j)}) - \left(1-y_j\right)\log (1-\F_{b\lambda}^{(j)}) \ .
\end{equation} 
We optimize the probe parameters by means of the gradient descent algorithm, where the binary cross-entropy, averaged on the training set, is used as loss function. Namely
\begin{gather}
	\lambda \to \lambda - \frac{\eta_\lambda}{M} \sum_{j=0}^{M-1} \partial_\lambda H\left(y_{j},\F^{(j)}_{b\lambda}\right) \ , \\
	b \to b - \frac{\eta_b}{M} \sum_{j=0}^{M-1} \partial_b H\left(y_{j},\F^{(j)}_{b\lambda}\right) \ ,
\end{gather}
with $\eta_\lambda, \eta_b$ the learning rates of the probe and bias parameters, respectively. The derivatives with respect to the parameters and the bias yield
\begin{gather}
	\partial_\lambda H = \left(\partial_{\F} H \right)\left(\partial_{\xi} \sigma \right) \partial_\lambda \f \ , \\
	\partial_b H = \left(\partial_{\F} H \right)\partial_{\xi} \sigma \ , 
\end{gather}
with $\xi_{b\lambda} = \f_\lambda + b$. Then,
\begin{gather}
	\partial_{\F} H = \frac{\F-y}{\F(1-\F)} 
	\label{eq:EntropyDerivative} \ , \\
	\partial_{\xi}\sigma = \beta \F (1 - \F) \ ,
\end{gather}
with $\beta$ the hyperparameter of the sigmoid. For any complex function of real variable $h: \mathbb{R} \rightarrow \mathbb{C}$, it follows that $\partial_\lambda \left|h(\lambda)\right| = \Re\left[h(\lambda)(\partial_\lambda h(\lambda))^*\right]/\left|h(\lambda)\right|$. Hence, 
\begin{equation}
	\partial_\lambda \f = 2 \Re \left[ \langle \I_{\omega}, \U_{\omega} \rangle \langle \I_{\omega}, \partial_{\lambda}\U_{\omega} \rangle^* \right] \ .
	\label{eq:DerivativeCompleteApp}
\end{equation}	
Neglecting the phase of $\langle \I, \U \rangle$,
\begin{equation}
	\partial_\lambda \f \simeq 2 \sqrt{f} \Re \left[ \langle \I_{\omega}, \partial_{\lambda}\U_{\omega} \rangle \right] \ .
	\label{eq:DerivativeApproximation}
\end{equation}
This assumption, which we verified in our simulations under a self-consistency test, simplifies the computation of the first factor of \cref{eq:DerivativeCompleteApp}, which is directly determined at the output of the Hong-Ou-Mandel interferometer.

\section{CLASSIFICATION IN THE FOURIER DOMAIN\label{app:Fourier}}
In this section, we discuss the effect of adding a single lens in the probe branch of the Hong-Ou-Mandel interferometer. We summarize the main calculations.

A thin lens is placed at one focal length $\ell$ from both the probe image plane and the beam splitter. In the near-field limit, the lens performs a Fourier transform of the probe state \citep{art:Rezai}, yielding $\ket{\Psi_\U} \to \ket{\Psi_{\U'}}$, where
\begin{equation}
	\U'_\omega(r|\lambda) = -i \frac{\omega}{\ell} e^{2i \omega f} \hat{\U}_{\omega}\left(\frac{\omega}{\ell}r\big|\lambda\right) \ .
	\label{eq:Lens}
\end{equation}
After the beam splitter, the rate of coincidences is
\begin{gather}
    p(1_a \cap 1_b|\lambda, \O) = \frac{1}{2}\left[\alpha_\omega(\O) - \widetilde{\f}_\lambda(\O) \right] 
    \label{eq:BosonicCoincidencesFourier} \ , \\
	\widetilde{\f}_\lambda(\O) =  \left| \langle \I_\omega(\cdot|\O),\hat{\U}_\omega (\cdot|\lambda) \rangle \right|^2 \ ,
	\label{eq:NetworkBraketFourierInProbe}
\end{gather}
yielding
\begin{gather}
	\widetilde{\F}_{b\lambda}(\O) = \sigma(\widetilde{\f}_{\lambda}(\O) + b) \ , \label{eq:FinalOutputFourier} \\
	\widetilde{\f}_\lambda(\O) = \left| \int_I \dif r \dif r' \ \I_\omega(r|\O)\U_{\omega}^*(r'|\lambda)e^{i r\cdot r'} \right|^2 \ , \label{eq:NetworkFourier}	
\end{gather}
with $\sigma$ and $b$ the sigmoid activation function and bias, already introduced in the main discussion. However, $\widetilde{\f}_\lambda(\O)$ is not a point-wise evaluation: it combines the image spatial modes with the momentum spectrum of the probe state. Using the duality of the Fourier transform, it follows that
\begin{equation}
	\widetilde{\f}_\lambda(\O) =  \left| \langle \hat{\I}_\omega(\cdot|\O),\U_\omega (\cdot|\lambda) \rangle \right|^2 \ ,
	\label{eq:NetworkBraketFourierInImage}
\end{equation}
which corresponds to the output illustrated above, but with the thin lens placed in the left branch, before the beam splitter. Equivalently, this takes the Fourier transform of the image, instead of that of the probe. In the next section, we leverage this symmetry to simplify both the training process and the numerical simulations.

The training of the model follows the same procedure of the previous section. By placing the lens on the top branch of the interferometer, while using the duality of the Fourier transform, we get
\begin{equation}
	\partial_\lambda \widetilde{\f} \simeq 2 \sqrt{\widetilde{\f}} \Re \left[ \langle \hat{\I}_{\omega}, \partial_{\lambda}\U_{\omega} \rangle \right] \ .
	\label{eq:DerivativeApproximationFourier}
\end{equation}
Under the same conditions of the spatial domain, the last two equations become
\begin{gather}
	\widetilde{\f}_\lambda(\O) = \left| \sum_{\mu,\nu} (u \star \hat{\I}^*_{\omega})(r_{\mu\nu}) \frac{\lambda_{\mu \nu}}{||\lambda||} \right|^2 
	\label{eq:LCDLayerFourier} \ , \\
	\partial_{\mu \nu} \widetilde{\f} \simeq 2 \frac{\sqrt{\widetilde{\f}}}{||\lambda||}\Re\left[ (u \star  \hat{\I}^*_{\omega})(r_{\mu\nu}) - \sqrt{\widetilde{\f}} \frac{\lambda_{\mu \nu}}{||\lambda||} \right] \  ,
	\label{eq:LCDDerivativeFourier}
\end{gather}
where in the last step we neglected the phase of $\langle \hat{\I},\U \rangle$. Here, \cref{eq:LCDDerivativeFourier} can be evaluated in an all-optical way through the characterization of the real part of $\hat{\I}$, namely, by performing an amplitude and phase measurement at the output of a thin lens, placed in the left branch, before the beam splitter. In our simulations, we compare the predictability of the neuron in the spatial and Fourier domains.

\section{LITERATURE COMPARISON\label{app:LiteratureComparison}}
In the literature, the Hong-Ou-Mandel effect has been applied to quantum kernel evaluation \citep{art:Bowie}, bringing a high-quality and a valuable contribution to the quantum machine learning research. We highlight the unique advantages and differences between those results and our method.

\begin{enumerate}
	\item Both methods describe a learning model, and rely on the Hong-Ou-Mandel effect to compute its prediction. A major difference concerns the states fed into the branches of the interferometer. In \citep{art:Bowie}, a single branch contains both the data and the parameters, which have to be encoded by optical state preparation. Each branch corresponds to a data point, while the Hong-Ou-Mandel interferometer takes the scalar product between them. In our method, input and parameters are fed as two separate branches, i.e. the input branch and the probe branch. In this case, the purpose of the Hong-Ou-Mandel is twofold. On the one hand, it performs the same scalar product of \citep{art:Bowie}, evaluating the distance between the input and the probe states. On the other hand, it automatically injects the parameters into the input data, bypassing the cost of directly encoding them. For this reason, the input and the probe can be treated, and encoded, separately.
	\item The two methods have different computational meaning. Reference \citep{art:Bowie} is an optical implementation of a kernel method, whose purpose is to compare data by learning the optimal representation that separates them. Our proposal implements a single neuron (one branch for the input, the other one for the parameters), whose purpose is to learn ad adapt to the data. Our method is also prone to further generalization: it represents the building block towards a new class of optical neural networks with universal approximation capabilities.
	\item In \citep{art:Bowie}, data are represented with temporal modes. Instead, we focus on spatial modes, independently of the type of encoding. State preparation is not needed as both the input and the parameters are loaded through free space propagation. Additionally, we use the SLMs to describe an explicit example and implementation of our model, discussing its training with minimal optical operations.
	\item Reference \citep{art:Bowie} claims an exponential advantage, concerning the number of computational resources saved according to \citep{art:Mohseni}. Our claim is different. We demonstrate a superexponential speedup, for computational and optical resources. Our method can classify images of any resolution with constant number of photons and mathematical operations. Our conclusions are supported by a detailed analysis of the resource cost when attempting the same task by classical means.
\end{enumerate}

\section{OPTICAL AND COMPUTATIONAL ADVANTAGE\label{app:Advantage}}
In this section, we discuss the optical and computational advantage as the number of photons and operations required by a single image classification. Assuming that all the parameters have been previously trained with optimal accuracy, we show that our protocol requires a constant number of resources, i.e. $\mathcal{O}(1)$ complexity, independently of the input image resolution: it provides a superexponential speedup over its classical counterpart. 

We first discuss the computational advantage when substituting a classical neuron with a quantum optical one. From now on, we denote $\Omega$, $\Theta$ and $\mathcal{O}$, respectively the lower, tight and upper bounds on the number of resources needed by a certain (optical or computational) operation. Consider a digital image $x$ of $N$ pixels, fed into a neuron
\begin{equation}
	G_{bw}(x) = \sigma(w\cdot x + b) \ ,
	\label{eq:ComputationalNeuronApp}
\end{equation}
where $x, w \in \mathbb{R}^N$, $b \in \mathbb{R}$ and $\sigma$ is the sigmoid  activation, with hyperparameters $\beta = 1$ and $\gamma = 0$. \cref{eq:ComputationalNeuronApp} costs $N$ operations to compute $w \cdot x$. The Hong-Ou-Mandel interferometer performs the same operation in an all-optical way, leaving the computational cost of the activation function and bias only, which is $\mathcal{O}(1)$.

We now discuss the optical advantage when using coincidences to classify single-photon states instead of a classical neuron on fully reconstructed images. After targeting an object with light, a digital image $x$ is an ensemble of grey levels obtained by counting the number of photons collected by different pixels on a sensor grid, e.g. a charge-coupled device \citep{art:Boyle}. Let $n_{p}$ be the average number of photons in the input state, and $\mu_i$ the average number of photons collected by the $i$-th pixel of the grid, with $i \in \{0,\ldots,N\}$. Assuming perfect quantum efficiency and sufficiently low exposure times to neglect the saturation of the sensor, the grey values at each pixel read
\begin{equation}
	x_i =  \frac{\mu_iL}{\mu_w} \ ,
\end{equation}
with $L$ the number of grey levels, i.e. the depth of the image, and $\mu_w = \max_{i}\mu_i$ the maximum number of photons collected in a single pixel. Indeed, $x_i \in \{0, 1, \ldots, L-1\}$ with $0$ and $L-1$ labelling the black and white colors, respectively. Each pixel has variance $\varsigma_i^2 = \Delta \mu_i^2 L^2/\mu_w^2$, with $\Delta\mu_i^2$ the variance on the number of collected photons. For coherent light, the photo-detection process undergoes the standard quantum limit (SQL) \citep{art:Kolobov, art:Fabre}, with Poissonian fluctuations that satisfy $\Delta\mu_i^2 \simeq \mu_i$. The average uncertainty reads
\begin{equation}
	\varsigma^2 := \frac{1}{N}\sum_{i=0}^{N-1} \varsigma_i^2 \stackrel{\text{SQL}}{\simeq} \langle x \rangle^2  N n_{p}^{-1} \ ,
	\label{eq:AvgStdDev}
\end{equation}
with $\langle x \rangle = N^{-1}\sum_i x_i \in [0,L-1]$ the average brightness of the image, which we assume to be independent of its resolution. Hence, the number of photons $n_p$ required by a full image reconstruction with average variance $\varsigma^2$ is $\varsigma^{-2}\langle x \rangle^2 N$. By defining the signal-to-noise ratio as $\eta = \langle x \rangle / \varsigma$, the number of photons needed simplifies to $\eta N$. This discussion involves the cost of image reconstruction only. We now take into account the information propagation through the neuron of \cref{eq:ComputationalNeuronApp}.
\begin{proposition}
Consider a neuron with sigmoid activation function. Suppose that there exists a sequence of parameters $\{ (w_{N} , b_{N}) \in \mathbb{R}^{N+1} \}_{N \gg 1}$ that optimally solve the $N$-pixel image classification task, with $b_N$ and the $\ell^1$-norm $||w_N||_1$ asymptotically bounded for $N \to \infty$. Then, the number of photons $n_p$ required to classify an image $x$ with uncertainty $\varepsilon$, is $\Omega\left(\varepsilon^{-2}\langle x \rangle N\right)$.
\end{proposition}
\begin{proof}
Consider the output of the neuron $G_{bw}(x) = \sigma(w_{N}\cdot x + b_N)$, and its derivative $\partial G(x) = G_{bw}(x)(1-G_{bw}(x))$. By neglecting the spatial neighbourhood correlations, which may introduce at most a constant overhead in our estimation, we propagate the uncertainty of $x$ as
\begin{equation}
	\varepsilon^2 = \langle x \rangle (\partial G)^2(x) \sum_{i=0}^{N-1}(w_N)_i^2 x_i N\tilde{n}^{-1}_p \ ,
	\label{eq:StandardErrorClassification}
\end{equation}
where $\tilde{n}_p = n_r n_p$, with $n_r$ is the number of independent image acquisition and classification. Since black pixels do not contribute to this summation, we get 
\begin{equation}
	\sum_{i=0}^{N-1}(w_N)_i^2 x_i \geq \sum_{i\notin\mathcal{B}}(w_N)_i^2 = ||w_N||^2 - \sum_{i\in\mathcal{B}}(w_N)_i^2 \ ,
\end{equation}
with $\mathcal{B} = \{i\in \mathbb{N} \ | \ x_i = 0 \ \text{for} \ 0 \leq i \leq N - 1 \}$ the set of black pixels labels. However, $||w_N||^2 > \sum_{i\in\mathcal{B}}(w_N)_i^2$. Otherwise, $||w_N||^2 \simeq \sum_{i\in\mathcal{B}}(w_N)_i^2$ would imply either that the image is mostly black, independently of its resolution, or that $(w_N)_i \simeq 0$ for all non-black pixels, which are both conditions that prevent the learnability of the neuron. By substitution into \cref{eq:StandardErrorClassification} we get
\begin{equation}
	\tilde{n}_p \geq \varepsilon^{-2} \langle x \rangle (\partial G)^2(x) ||w_N||^2 N \ .
\end{equation}
Since $w_N$ is a sequence of non-trivial solutions of the classification problem, the $\ell^2$-norm $||w_N||^2$ cannot go to zero for $N \to \infty$. Finally, we show that $(\partial G(x))^2$ does not converge to $0$ for $N \to \infty$. Consider
\begin{equation}
	(\partial G)^2(x) = \frac{e^{-2(w_{N}\cdot x + b_N)}}{[1+e^{-(w_{N}\cdot x + b_N)}]^4} \ .
\end{equation}
If $b_{N}$ is asymptotically limited, $(\partial G)^2$ converges to zero if and only if $w_{N}\cdot x \to \pm \infty$. By splitting this scalar product into positive and negative contributions $w_{N}\cdot x = \sum_{(w_{N})_i > 0}(w_{N})_i x_i - \sum_{(w_{N})_i < 0}|(w_{N})_i| x_i$, it follows that
\begin{align}
	&w_{N}\cdot x \leq \sum_{(w_{N})_i > 0}(w_{N})_i x_i \leq L||w_N||_1 \ , \\
	&w_{N}\cdot x \geq - \sum_{(w_{N})_i < 0}|(w_{N})_i| x_i \geq - L||w_N||_1 \ ,
\end{align}
namely that $|w_{N}\cdot x| \leq L ||w_N||_1$. Since the $\ell^1$-norm is limited, $(\partial G)^2$ admits strictly positive lower bound for $N \to \infty$. Finally, this imply that $\tilde{n}_{p} = \Omega\left(\varepsilon^{-2}\langle x \rangle N\right)$.
\end{proof}
In the previous discussion, two conditions lead to the above lower bound. On the one hand, that $||w_N||^2 \not\to 0$ for $N \to \infty$, which is essential to guarantee that the neuron is trainable at any resolution. On the other hand, that $||w_N||_1$ is bounded for $N \to \infty$, which is compatible with LASSO and Tikhonov's regularization techniques \citep{art:Santosa,art:Tikhonov}.

Our protocol exponentially reduces this cost, by performing the measurement with two bucket detectors only, with no spatial resolution. The object classification depends on the estimation of the rate of coincidences of the Hong-Ou-Mandel interferometer. Let $\tilde{n}_p = 2n_p$ be the number input photons, and $\tilde{p} \in [0,1/2]$ the empirical rate of coincidences. Under the normal approximation, with the $95\%$ confidence level \citep{book:Rotondi}, the estimation uncertainty reads
\begin{equation}
	\varepsilon = 2 \sqrt{\frac{\tilde{p}(1-\tilde{p})}{\tilde{n}_p}} \ .
\end{equation}
Since $4 \tilde{p}(1-\tilde{p}) \leq 1$, the total number of photons is $\O(\varepsilon^{-2})$, which is constant with respect to the resolution of the image. In conclusion, the quantum optical neuron provides a superexponential advantage over its classical counterpart, both in the number of operations and photons saved to classify a single image. We summarize these results and the full discussion in Supplementary Table I.
\begin{table}[H]
	\centering
	\def\arraystretch{1.5}
	\setlength\tabcolsep{5pt}
	\begin{tabular}{|c|c|c|c|}
		\hline
		\multicolumn{2}{|c|}{Resources} & \multicolumn{1}{c|}{\parbox{3cm}{\centering \ \\[0.5pt] Quantum optical \\ neuron \\[4pt]}} & \parbox{3cm}{\centering \ \\[0.5pt] Imaging $\to$ Classical \\ artificial neuron \\[4pt]} \\ \hline
		\multicolumn{2}{|c|}{\parbox{3cm}{\centering \ \\[0.5pt] Computational \\ (\# of mathematical operations) \\[3pt]}} & $\mathcal{O}(1)$ & $N$ \\ \hline 
 		\multirow{2}{*}{\parbox{2.25cm}{\centering Optical \\ (\# of photons)}} & Imaging & \parbox{3cm}{\centering \ \\[0.5pt] The image is not \\ reconstructed \\[4pt]} & $\varsigma^{-2}\langle x \rangle^2 N$ \\ \cline{2-4}
 		& Classification & $\mathcal{O}(\varepsilon^{-2})$ & $\Omega(\varepsilon^{-2}\langle x \rangle N)$ \\ \hline
	\end{tabular}
	\caption{\label{tab:FullResourceCost}Extended comparison of the computational and optical resources needed to classify an image $x$ of $N$ pixels. (Computational) An artificial neuron requires $N$ mathematical operations to inject the parameters. Our method performs the same task by optical means, through a single Hong-Ou-Mandel interferometer. No computation is required other than scalar post-processing. (Optical) In coherent imaging, the number of photons needed depends linearly on the output resolution, i.e. $\varsigma^{-2}\langle x \rangle^2 N$. Here, $\varsigma$ and $\langle x \rangle $ are the standard deviation and the average brightness of the reconstructed image (which depends on the reflectivity of the object). The same also holds for the artificial neuron, after which the cost increases to $\Omega(\varepsilon^{-2}\langle x \rangle N)$, where $\varepsilon$ is the uncertainty on the classification outcome, i.e. the error on the predicted class. Our method achieves the same task and uncertainty, with number of photons dictated only by the binomial statistics, independently of the resolution. In both cases, the quantum optical neuron requires constant resources, achieving a superexponential speedup.}
\end{table}

We conclude this section by comparing our method against a classical optical design. We replace the artificial neuron with a Mach-Zehnder interferometer, which measures the overlap between the input and the probe states using classical, e.g. coherent, light. Consider a multimode coherent state $\ket{\xi_\phi} = D(\xi,\phi)\ket{0}$, with intensity $|\xi_\phi|^2$ and spatial spectrum $\phi$. Here 
\begin{gather}
	D(\xi,\phi) = \exp(\xi A^\dagger_\phi - \xi^*A_\phi) \ , \\
	A^\dagger_\phi = \int \dif^2k \ \hat{\phi}(k) a^\dagger_\omega(k) \ , \label{eq:MultimodeCreationOperator}
\end{gather}
are the multimode displacement and creation operators, respectively, and $\exp$ denotes the exponential map. A $50\!:\!50$ beam splitter separates the source into two multimodes, $A$ and $B$. Using \cref{eq:BeamSplitterOperation}, we get
\begin{equation}
	\ket{\xi_\phi} \to \exp[\frac{1}{\sqrt{2}}\left(\xi A^\dagger_\phi - \xi^*A_\phi\right) + \frac{1}{\sqrt{2}}\left(\xi B^\dagger_\phi - \xi^*B_\phi\right) ] \ket{0} = \ket{\frac{\xi_\phi}{\sqrt{2}}}  \ket{\frac{\xi_\phi}{\sqrt{2}}} ,
	\label{eq:CoherentSplitting}
\end{equation} 
where $B$ undergoes the same definition of \cref{eq:MultimodeCreationOperator}, with respect to mode $b$. In this case, the beam splitter halves the source intensity across the output branches. Consider two SLMs, each with $N$ pixels, encoding the input image and the trainable parameters. In the simplest scenario, each pixel can be associated to a single spatial mode, by placing a lens at each branch of the interferometer. At one focal length, this performs a Fourier transform, which, for finite number of modes, reads as the unitary transformation 
\begin{equation}
	U = \frac{1}{\sqrt{N}}\sum_{ij}\zeta_{ij} e_i \otimes e_j \ ,
\end{equation}
with $|\zeta{ij}|^2 = 1 \ \forall i,j$, and $e_i \otimes e_j$ labelling the $ij$-th component of $U$ \citep{book:Nielsen}. Similarly to \cref{eq:BeamSplitterOperation,eq:CoherentSplitting}, the multimode $A$ transforms under the action of $U$, splitting each branch in $N$ spatial modes $\{B_1,\ldots,B_{N}\}$ and yielding
\begin{equation}
	\ket{\frac{\xi_\phi}{\sqrt{2}}} \to \ket{\frac{\xi_\phi}{\sqrt{2N}}}\ldots \ket{\frac{\xi_\phi}{\sqrt{2N}}} \ .
	\label{eq:CoherentFactorization}
\end{equation}
Encoding the SLM, i.e. having at least one photon for each pixel, requires a coherent state with $|\xi_\phi|^2 > 2N$, namely an average number of photons that scales as $\Omega(N)$. This means that a classical interferometer would bring a computational speedup without optical advantage, undergoing the same cost of image detection and subsequent classification. In our method, single-photon states do not factorize as in \cref{eq:CoherentFactorization}. They use entanglement to encode the SLM in a superposition of $N$ spatial mode, which, combined with the coincidence measurement, leads to the $\mathcal{O}(1)$ scaling.

\begin{figure}[H]
	\centering
	\includegraphics[width = 1 \textwidth]{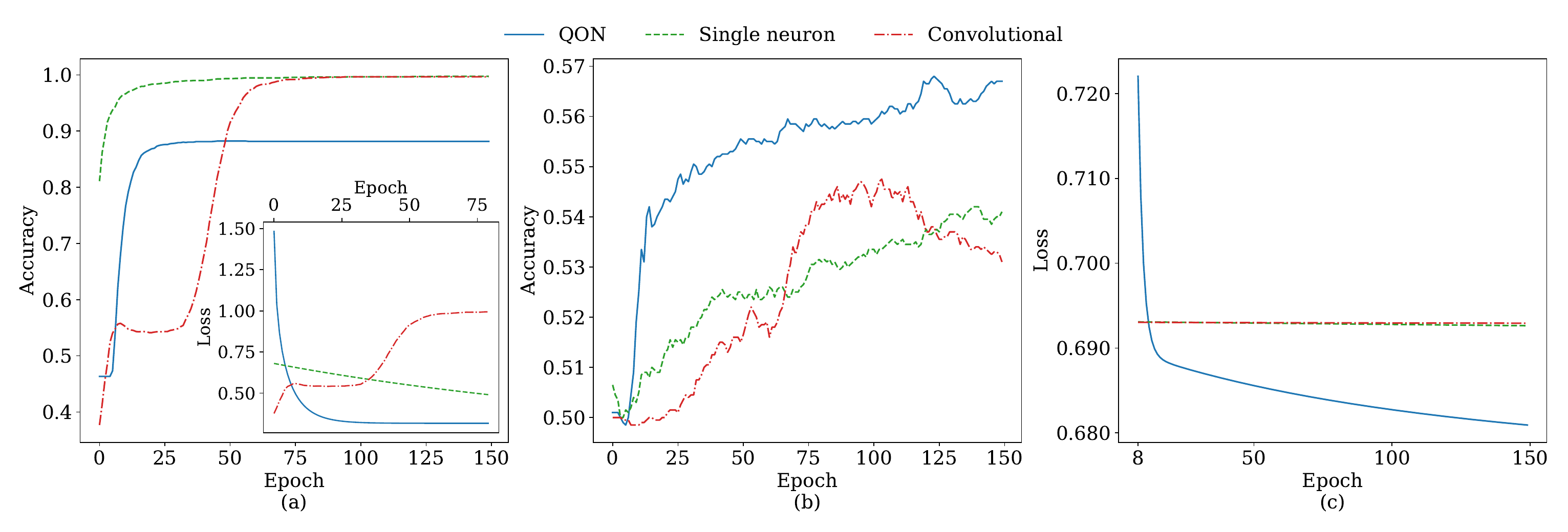}%
	\caption{\label{fig:SquareModulus}Quantum optical neuron with square absolute value activation function. Comparison between the quantum optical neuron (QON), both in the spatial (blue solid line) and in the Fourier (orange dotted line) domains, a classical artificial neuron (green dashed line) and a convolutional network (red dash-dotted line). All models are trained with the same number of $\sim 1024$ parameters, optimizer and $\eta_\lambda = 0.075$ and $\eta_b = 0.005$.  (a) History plot for the MNIST dataset, when classifying zeros and ones. Our method reaches an asymptotic accuracy near $88\%$. The inset shows the binary cross-entropy. (b-c) Accuracy and loss histories for the CIFAR-10 dataset, when classifying cats and dogs. Our method reaches an asymptotic accuracy near $57\%$.}
\end{figure}
\begin{figure}
	\centering
	\includegraphics[width = 1 \textwidth]{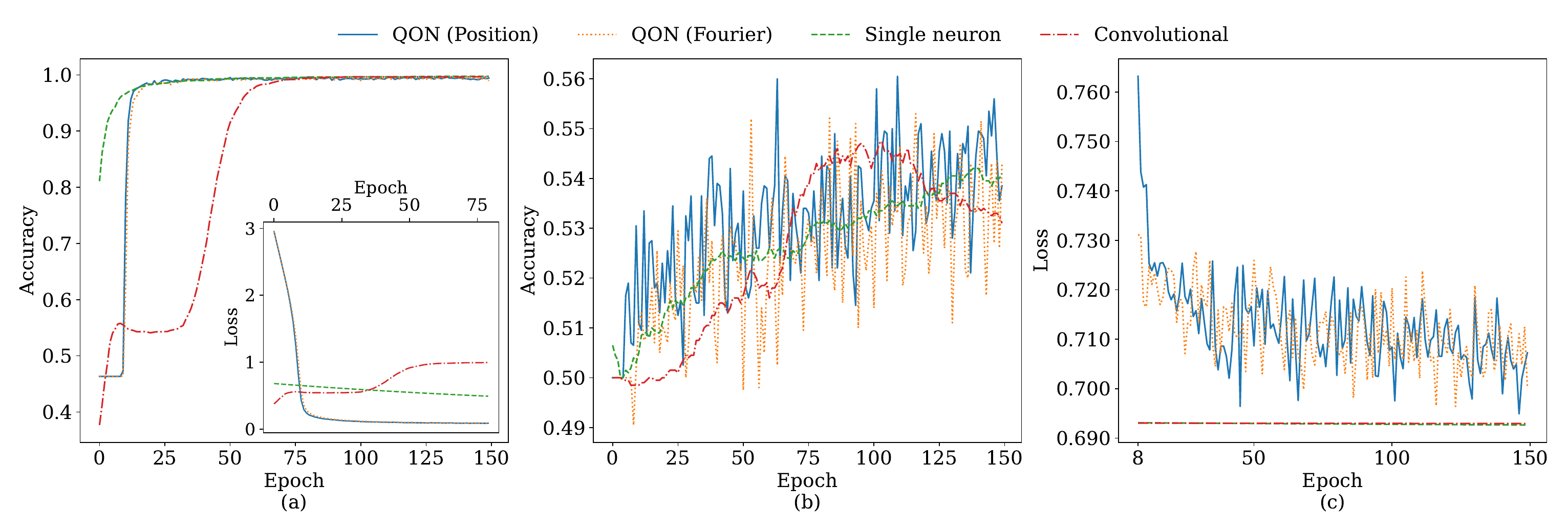}
	
	\includegraphics[width = 1 \textwidth]{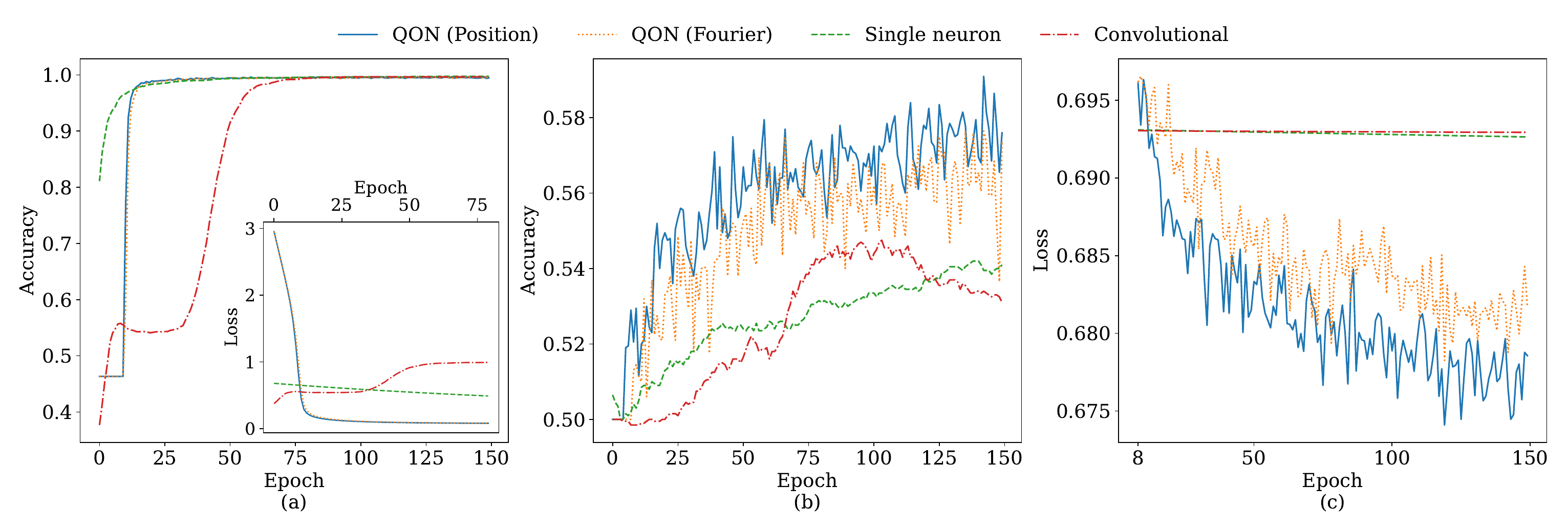}%
	\caption{\label{fig:Noise}Quantum optical neuron with coincidences approximated by empirical frequencies. Comparison between the quantum optical neuron (QON), both in the spatial (blue solid line) and in the Fourier (orange dotted line) domains, a classical artificial neuron (green dashed line) and a convolutional network (red dash-dotted line). All models are trained with the same number of $\sim 1024$ parameters, optimizer and $\eta_\lambda = 0.075$ and $\eta_b = 0.005$. Output sampled with $100$ (upper plots) and $1000$ (lower plots) repetitions. (a) History plot for the MNIST dataset, when classifying zeros and ones. The inset shows the binary cross-entropy. Our method demonstrates high noise resiliency. (b-c) Accuracy and loss histories for the CIFAR-10 dataset, when classifying cats and dogs. Although with lower resiliency, sampling fluctuations does not affect the asymptotic convergence of training.}
\end{figure}

\end{widetext}

\clearpage
\FloatBarrier
\bibliography{refs.bib}
\end{document}